\pdfoutput=1
\RequirePackage{ifpdf}
\ifpdf % We are running pdfTeX in pdf mode
\documentclass[pdftex]{sigma}
\else
\documentclass{sigma}
\fi

\usepackage[large,nohug,heads=vee]{diagrams}

\def\spisok#1{\begin{gather}#1\end{gather}}

\newcounter{trans}\renewcommand*\thetrans{T\arabic{trans}}
\newcounter{eqlist}\renewcommand*\theeqlist{E\arabic{eqlist}}

\begin{document}
\allowdisplaybreaks

\newcommand{\arXivNumber}{1903.11893}

\renewcommand{\PaperNumber}{062}

\FirstPageHeading

\ShortArticleName{Integrable Modifications of the Ito--Narita--Bogoyavlensky Equation}
\ArticleName{Integrable Modifications\\ of the Ito--Narita--Bogoyavlensky Equation}

\Author{Rustem N.~GARIFULLIN and Ravil I.~YAMILOV}

\AuthorNameForHeading{R.N.~Garifullin and R.I.~Yamilov}
\Address{Institute of Mathematics, Ufa Federal Research Centre, Russian Academy of Sciences,\\ 112 Chernyshevsky Street, Ufa 450008, Russia}
\Email{\href{mailto:rustem@matem.anrb.ru}{rustem@matem.anrb.ru}, \href{mailto:RvlYamilov@matem.anrb.ru}{RvlYamilov@matem.anrb.ru}}

\ArticleDates{Received April 01, 2019, in final form August 14, 2019; Published online August 23, 2019}

\Abstract{We consider five-point differential-difference equations. Our aim is to find integrable modifications of the Ito--Narita--Bogoyavlensky equation related to it by non-invertible discrete transformations. We enumerate all modifications associated to transformations of the first, second and third orders. As far as we know, such a~classification problem is solved for the first time in the discrete case. We analyze transformations obtained to specify their nature. A number of new integrable five-point equations and new transformations have been found. Moreover, we have derived one new completely discrete equation. There are a~few non-standard transformations which are of the Miura type or are linearizable in a~non-standard way. We have also proved that the orders of possible transformations are restricted by the number five in this problem.}


\Classification{37K05; 37K10; 35G20}

\section{Introduction}
We consider differential-difference equations of the form
\begin{gather}\dot v_n=f(v_{n+2},v_{n+1},v_n,v_{n-1},v_{n-2}),\label{eqvgen}\end{gather}
where $v_n=v_n(t)$ is an unknown function of the continuous time $t$ and discrete integer variable~$n$,~$\dot v_n$ denotes the time derivative of $v_n$, and $f$ is a function of five variables.
The oldest and the most famous integrable example of this class is the Ito--Narita--Bogoyavlensky (INB) equation~\cite{bo88, i75,na82}
\begin{gather}\dot u_n=u_n(u_{n+2}+u_{n+1}-u_{n-1}-u_{n-2})\label{INB}.\end{gather}

Lately, equations of this class have been intensively studied by the generalized symmetry method, see, e.g., \cite{a14,a16, GYL17,GYL18}. However, the problem of description of all integrable equations of this class is far from complete. To find new integrable equations, we apply in this article an alternative approach using non-invertible discrete transformations.

More precisely, we use non-invertible transformations of the following special form
\begin{gather}\label{tran_g}u_n=g(v_{n+k_1},v_{n+k_1-1},\ldots,v_{n+k_2+1},v_{n+k_2}),\qquad k_1> k_2,\end{gather} relating \eqref{eqvgen} and \eqref{INB}. They transform any solution $v_n$ of \eqref{eqvgen} into a solution $u_n$ of~\eqref{INB}. Such transformation is explicit in one direction and can be called the discrete substitution by analogy with differential substitutions in the continuous case, see, e.g., \cite{s88, s01}, where such transformations were studied. Numerous examples of transformations of the form~\eqref{tran_g} can be found in \cite{y06} for the Volterra and Toda type equations and in \cite{GYL17,GYL18} for the five-point equations~\eqref{eqvgen}. Different methods for the construction of discrete transformations \eqref{tran_g} are presented in~\cite{GYL16, y94}.

Equation \eqref{INB} is one of the key equations of lists of integrable equations found in \cite{GYL17,GYL18} as a result of the generalized symmetry classification of an important subclass of \eqref{eqvgen}. Here we are going to enumerate all modifications of the INB equation \eqref{INB}, which are associated with transformations \eqref{tran_g} of the orders $k=1,2,3$, where $k=k_1-k_2$. We will also prove that the order of a possible transformation in this problem is restricted by the number five: $k\leq 5$. This estimate is accurate in the sense that there exist transformations for all orders $1\leq k\leq 5$. Examples of transformations of any such an order will be given below. As far as we know, a~classification problem of this kind is solved for the first time in the discrete case, as for continuous case, see, e.g.,~\cite{kpz12}.

A well-known example of a transformation of the form \eqref{tran_g} is the discrete Miura transformation~\cite{w76}
\begin{gather} \label{miura}u_n=(v_{n+1}+1)(v_n-1)\end{gather} relating the Volterra equation and its modification.
There is a class of more simple linearizable transformations in the terminology of~\cite{GYL16}. Such transformations are more simple than Miura type ones in the sense that the problem of finding $v_n$ by~\eqref{tran_g}, starting from a given function~$u_n$, is more easy. We analyze transformations obtained in this paper to show that the most of them are linearizable.

As a result of the classification we obtain a number of new integrable equations and transformations. Most of new equations belong to the class~\eqref{eqvgen} and one of them is the discrete quad-equation. Among new transformations found here, one is of Miura type and two are linearizable in a non-standard way.

In Section \ref{theo} we discuss some theoretical aspects and prove a boundedness theorem for the orders of possible transformations. In Section~\ref{low} transformations of the orders~1 and~2 are classified together with corresponding modifications. Section~\ref{third} is devoted to transformations of the order~3. In Section~\ref{secN} we discuss in detail the most interesting examples of equations and transformations obtained in the previous section. In conclusion we briefly summarize results obtained in the paper.

\section{Theoretical comments and results}\label{theo}

The differential-difference equations \eqref{eqvgen} and discrete transformations~\eqref{tran_g} we consider in this paper are autonomous, i.e., do not explicitly depend on the discrete variable $n$. That is why, for brevity, we may write down equations and transformations at the point $n=0$:
\begin{gather}\dot v=f(v_2,v_1,v,v_{-1},v_{-2}), \label{eqv}\end{gather} where $v=v_0.$
Moreover, up to the shift of the discrete variable $n$, transformations~\eqref{tran_g} can be rewritten in the form
\begin{gather} u=g(v_k,v_{k-1},\ldots,v_1,v),\qquad k\geq 1,\label{tran}\end{gather}where $u=u_0$. We also use a natural restriction
\begin{gather}\frac{\partial g}{\partial v_k}\neq 0,\qquad \frac{\partial g}{\partial v}\neq 0.\label{usl-g}
\end{gather}

We are going to enumerate all modifications of the form \eqref{eqv} of the INB equation \eqref{INB}, which correspond to transformations \eqref{tran} of the orders $k = 1,2,3$. Such modifications are integrable equations in the sense that they possess infinitely many conservation laws.

In fact, the INB equation \eqref{INB} has conservation laws of an arbitrarily high order, which can easily be constructed by using the well-known Lax representation \cite{bo88} or recursive operator~\cite{ztof91}. Modifications~\eqref{eqv} are related to the INB equation by transformations of the form \eqref{tran} and, for this reason, also have conservation laws of an arbitrarily high order due to~\cite[Theorem~18]{y06}. That theorem is formulated for the Volterra type equations, but it can easily be reformulated for the case of five-point equations \eqref{eqv}. Conservation laws for a modification \eqref{eqv} are constructed in an explicit way by using the corresponding transformation \eqref{tran} and known conservation laws of \eqref{INB} \cite[Section~2.7]{y06}.

Let us introduce the notation $h_j$ for any function $h=h(v_{m_1},v_{m_1-1},\ldots,v_{m_2})$:
$h_j=T^j h,$ where $T$ is the shift operator defined by \begin{gather*}T^j h(v_{m_1},v_{m_1-1},\ldots,v_{m_2})=h(v_{m_1+j},v_{m_1-1+j},\ldots,v_{m_2+j})\end{gather*} for any integer $j$, in particular, $h_0=h$.

If equation \eqref{eqv} is transformed into \eqref{INB} by transformation \eqref{tran}, then the functions $f$,~$g$ must satisfy the determining equation
\begin{gather} D_t g=g(g_2+g_1-g_{-1}-g_{-2}), \qquad D_t g=\sum_{j=0}^k \frac{\partial g}{\partial v_j}f_j\label{usl_z},\end{gather}
i.e., here $D_t $ is the operator of differentiation in virtue of \eqref{eqv}.

Our main problem can formulated as follows: for any fixed $k\geq1$ we look for pairs of func\-tions~$f$,~$g$ satisfying~\eqref{usl_z}. It is important that the functions $v_j,\ j\in\mathbb{Z},$ are considered in this problem as the independent variables. The functions $f$,~$g$ must identically satisfy \eqref{usl_z} for all values of these variables.

\looseness=-1 Below we prove that the orders of possible transformations in this problem are restricted by the number five: $k\leq 5$. This estimate is accurate in the sense that there exist transformations for all the orders $1\leq k\leq 5$. Examples of transformations of the orders $k=1,2,3$ will be given in Sections~\ref{low} and~\ref{third}, while examples of the orders $k=4,5$ are presented in both Sections~\ref{secN1} and~\ref{secN2}.

\begin{theorem}\label{order} An equation of the form \eqref{eqv} cannot be transformed into the INB equation \eqref{INB} by a transformation of the form \eqref{tran} for any order $k>5$.
\end{theorem}

\begin{proof} Let us rewrite condition \eqref{usl_z} as \begin{gather}D_t \log g=g_2+g_1-g_{-1}-g_{-2}.\label{usl_z1}\end{gather}
Differentiating \eqref{usl_z1} with respect to $v_{k+2}$ and dividing the result by $\frac{\partial f_k}{\partial v_{k+2}}$, one has
\begin{gather}\frac{\partial \log g}{\partial v_k}=\frac{\partial g_2}{\partial v_{k+2}}\Big/\frac{\partial f_k}{\partial v_{k+2}}\label{sl_u}.\end{gather}
We see that, if $k\geq 4$, then the right hand side of this relation and therefore the function $\frac{\partial \log g}{\partial v_k}$ do not depend on $v_1$,~$v$. Hence, the function $g$ can be represented in the form
\begin{gather}g=g^{(1)}(v_k,v_{k-1},\ldots,v_2)g^{(2)}(v_{k-1},\ldots,v_1,v).\label{vid_g}\end{gather}

Let us divide \eqref{sl_u} by $g_2^{(2)}$
\begin{gather} \frac{1}{g^{(2)}_2}\frac{\partial \log g^{(1)}}{\partial v_k}=\frac{\partial g_2^{(1)}}{\partial v_{k+2}}\Big/\frac{\partial f_k}{\partial v_{k+2}}\label{sl1_u}.\end{gather}
If $k\geq 6$, then both sides of this relation depend on $v_{k+1},v_k,\ldots,v_4$ only, therefore
\begin{gather*} \frac{\partial \log g^{(1)}}{\partial v_k}=g^{(2)}_2 f^{(1)}(v_{k+1},v_k,\ldots,v_4),\end{gather*} where $f^{(1)}$ is a new function. As the left hand side of \eqref{sl1_u} does not depend on $v_{k+1}$, we apply the operator $T^{-2}\frac{\partial }{\partial v_{k+1}}\log$ and derive the following consequence
\begin{gather} \frac{\partial \log g^{(2)}}{\partial v_{k-1}}=-\frac{\partial \log f^{(1)}_{-2}}{\partial v_{k-1}}=f^{(2)}(v_{k-1},\ldots,v_2)\label{sl3_u}\end{gather} with a new function $f^{(2)}$.

Differentiating the main equation \eqref{usl_z1} for $g$ and $f$ with respect to $v_{k+1}$, we get
\begin{gather*} \frac{\partial \log g}{\partial v_k}\frac{\partial f_k}{\partial v_{k+1}}+\frac{\partial \log g}{\partial v_{k-1}}\frac{\partial f_{k-1}}{\partial v_{k+1}}=\frac {\partial g_2}{\partial v_{k+1}}+\frac{\partial g_1}{\partial v_{k+1}}.\end{gather*}
Then we differentiate the result with respect to $v_1$
\begin{gather*}\frac{\partial ^2g_1}{\partial v_{k+1}\partial v_1}=\frac{\partial f_k}{\partial v_{k+1}}\frac{\partial ^2\log g}{\partial v_k\partial v_1}+\frac{\partial f_{k-1}}{\partial v_{k+1}}\left(\frac{\partial ^2\log g^{(1)}}{\partial v_{k-1}\partial v_1}+\frac{\partial ^2\log g^{(2)}}{\partial v_{k-1}\partial v_1}\right).\end{gather*}
Here $\frac{\partial ^2\log g}{\partial v_k\partial v_1}=0$ due to \eqref{vid_g}, $\frac{\partial \log g^{(1)}}{\partial v_1}=0$ due to the definition of $g^{(1)}$ in \eqref{vid_g}, and $\frac{\partial ^2\log g^{(2)}}{\partial v_{k-1}\partial v_1}=0$ due to~\eqref{sl3_u}. Therefore $\frac{\partial ^2g}{\partial v_k\partial v}=0$, i.e.,
\begin{gather*}\frac{\partial ^2g}{\partial v_k\partial v}=\frac{\partial g^{(1)}}{\partial v_k}\frac{\partial g^{(2)}}{\partial v}=\frac{1}{g}\frac{\partial g}{\partial v_k}\frac{\partial g}{\partial v}=0,\end{gather*} but the last equality contradicts the conditions \eqref{usl-g}.
\end{proof}

\section{Modifications of the INB equation of the levels 1 and 2}\label{low}

We classify here modifications of the INB equation of the levels 1 and 2, which correspond to transformations of the orders 1 and 2. In the classification, we use differential consequences of the determining equation \eqref{usl_z} for the functions $f$ and $g$, as in the proof of Theorem \ref{order}, and we are trying all possible cases.

If necessary, we can change $v$ in an equation \eqref{eqv} and corresponding transformation \eqref{tran} by using the point transformation
\begin{gather}\tilde v=\vartheta(v)\label{point}.\end{gather} The transformation \eqref{tran} and the time $t$ in~\eqref{eqv} can be changed by the transformation
\begin{gather} \tilde u=\eta u,\qquad \tilde t=t/\eta\label{scale}\end{gather} leaving the INB equation \eqref{INB} invariant.
The classification will be carried out up to these autonomous point transformations \eqref{point} and \eqref{scale}.

An exact statement of the result will be given later. First of all, let us enumerate all possible modifications of the first level.

\medskip

{\centerline{{\bf List 1.} Modifications of the first level.}}
\vspace{-7mm}

\spisok{
\dot v=v(v_2v_1-v_{-1}v_{-2}),\stepcounter{eqlist}\tag{\theeqlist}\label{eqlist1}\\
\dot v=v\left(\frac{v_2}{v_1}+2\frac{v_1}v+2\frac{v}{v_{-1}}+\frac{v_{-1}}{v_{-2}}\right)+cv,\stepcounter{eqlist}\tag{\theeqlist}\label{eqlist2}\\
\dot v=v\big(v_2v_1^\kappa-\kappa^2v_1v^\kappa-vv_{-1}^\kappa+\kappa^2v_{-1}v_{-2}^\kappa\big),\stepcounter{eqlist}\tag{\theeqlist}\label{eqlist_nu}\\
\dot v=(v_{2}-v_{1}+a)(v-v_{-1}+a)+(v_1-v+a)(v_{-1}-v_{-2}+a)\nonumber\\
\hphantom{\dot v=}{}+(v_1-v+a)(v-v_{-1}+a)+c, \stepcounter{eqlist}\tag{\theeqlist}\label{eqlist4}\\
\dot v=\phi(v_2-v_1)\phi(v-v_{-1})+\phi(v_1-v)\phi(v_{-1}-v_{-2})\nonumber\\
\hphantom{\dot v=}{}+\phi(v_1-v)\phi(v-v_{-1})+c.\stepcounter{eqlist}\tag{\theeqlist}\label{eqlist5}
}

In equation \eqref{eqlist_nu} and in all the lists below
\begin{gather} \kappa=\big(1\pm {\rm i}\sqrt3\big)/2\label{eqkappa},\end{gather} i.e., $\kappa^3=-1$.
That is why one has in \eqref{eqlist_nu} two cases corresponding to the signs $+$ and $-$.
In equation \eqref{eqlist5} and in all the lists below, the function $\phi$ satisfies the differential equation
\begin{gather}\phi'=\frac{\phi-1}{\phi}\label{eqphi}.\end{gather} This function $\phi$ can be defined as
\begin{gather*}\phi(x)=1+\psi^{-1}(x+a), \qquad \psi(u)=u+\log u.\end{gather*}

Here and below $a$ and $c$ are arbitrary constants.
The constant $c$ indicates the existence of a~point symmetry. For example, equation~\eqref{eqlist2} has the point symmetry $v_\tau=v$, while equations~\eqref{eqlist4} and~\eqref{eqlist5} have the point symmetry $v_\tau=1.$ Let us remind that these point symmetries correspond to the one-parameter groups of auto-transformations $v\rightarrow {\rm e}^\tau v$ and $v\rightarrow v+\tau$, respectively.

The constants $a$ and $c$ can be removed by using one of the following non-autonomous $n$- and $t$-dependent point transformations
\begin{gather*}v=\tilde v {\rm e}^{ct},\qquad v=\tilde v+ct, \qquad v=\tilde v-an,\end{gather*} where $\tilde v$ is a new unknown function. The same is true for equations of Lists 2 and 3 below, except for equation~\eqref{eqlist_3_6}. In equations \eqref{eqlist_3_4}--\eqref{eqlist_3_81} of List 3, we can make $a=1$ by the transformation $z=\tilde z a^{-n}.$

Transformations corresponding to the equations of List 1 are presented in List 1$'$.

\medskip

{ \centerline{{\bf List 1$'$.} Transformations of the first order.}}
\vspace{-7mm}

\spisok{
u=v_1v,\stepcounter{trans}\tag{\thetrans}\label{trans1}\\
u=\frac{v_1}v,\stepcounter{trans}\tag{\thetrans}\label{trans2} \\
u=v_1v^\kappa,\stepcounter{trans}\tag{\thetrans}\label{trans_nu} \\
u=v_1-v+a,\stepcounter{trans}\tag{\thetrans}\label{trans3} \\
u=\phi(v_1-v)-1.\stepcounter{trans}\tag{\thetrans}\label{trans4}
}
Here and in all the lists below, a transformation with the number (TX) corresponds to the equation with the number (EX) for any X.

\begin{theorem}\label{ord1} If an equation \eqref{eqv} is transformed into the INB equation \eqref{INB} by a transformation~\eqref{tran} with $k=1$, then up to the autonomous point transformations~\eqref{point} and~\eqref{scale} it coincides to one of equations \eqref{eqlist1}--\eqref{eqlist5} of List~{\rm 1}. Equations \eqref{eqlist1}--\eqref{eqlist5} are transformed into the INB equation by transformations \eqref{trans1}--\eqref{trans4}, respectively.
\end{theorem}

Equations \eqref{eqlist1} and \eqref{eqlist4} together with transformations \eqref{trans1} and \eqref{trans2} are known, see \cite[List~4]{GYL17}, \cite[List~3]{GYL18} and references therein. Transformation \eqref{trans_nu} was discussed in \cite[Section~2.2]{GYL16}. All these transformations \eqref{trans1}--\eqref{trans4} are linearizable in the terminology of~\cite{GYL16}. Such transformations can be represented as compositions of point transformations and of linear transformations with constant coefficients. For instance, in the case of transformations \eqref{trans1}--\eqref{trans_nu}, one has
\begin{gather*} \hat u=\hat v_1+\alpha\hat v,\qquad u={\rm e}^{\hat u},\qquad v={\rm e}^{\hat v}.\end{gather*}

In the next list we enumerate all possible modifications of the INB equation of the second level.

\medskip

{ \centerline{{\bf List 2.} Modifications of the second level.}}
\vspace{-5mm}

\spisok{
\dot w=w\left(\frac{w_2}{w}+\frac{w_1}{w_{-1}}+\frac{w}{w_{-2}}\right)+cw, \stepcounter{eqlist}\tag{\theeqlist}\label{eqlist_2_1}\\
\dot w=w^2(w_2w_1-w_{-1}w_{-2}), \stepcounter{eqlist}\tag{\theeqlist}\label{eqlist_2_3}\\
\dot w=w\big(w_2w_1^{1+\kappa}w^\kappa-\kappa w_1w^{1+\kappa}w_{-1}^\kappa+\kappa^2 ww_{-1}^{1+\kappa}w_{-2}^\kappa\big),
\stepcounter{eqlist}\tag{\theeqlist}\label{eqlist_2_kappa1}\\
\dot w=w\big( w_2w_1^{\kappa-1}w^{-\kappa}+\big(1-\kappa^2\big)w_1w^{\kappa-1}w_{-1}^{-\kappa}-\kappa^2ww_{-1}^{\kappa-1}w_{-2}^{-\kappa}\big)+cw,
\stepcounter{eqlist}\tag{\theeqlist}\label{eqlist_2_kappa2}\\
\dot w=(w_1+w)(w+w_{-1})\left(w_{2}+w_1-w_{-1}-w_{-2}\right),
\stepcounter{eqlist}\tag{\theeqlist}\label{eqlist_2_+}\\
\dot w=(w_1-w+a)(w-w_{-1}+a)\left(w_{2}-w_1+w_{-1}-w_{-2}+2a\right)+c,
\stepcounter{eqlist}\tag{\theeqlist}\label{eqlist_2_-}\\
\dot w=\phi(w_1-w)\phi(w-w_{-1})\left(\phi(w_{2}-w_1)+\phi(w_{-1}-w_{-2})-1\right)\nonumber\\
\hphantom{\dot w=}{} -\phi(w_2-w_1)\phi(w-w_{-1})-\phi(w_1-w)\phi(w_{-1}-w_{-2})+c,
\stepcounter{eqlist}\tag{\theeqlist}\label{eqlist_2_fun}\\
\dot w=(w_1-w)(w-w_{-1})\left(\frac{w_2}{w_1}-\frac{w_{-2}}{w_{-1}}\right),
\stepcounter{eqlist}\tag{\theeqlist}\label{eqlist_2_sl}\\
\dot w=w(w+1)(w_2w_1-w_{-1}w_{-2})
\stepcounter{eqlist}\tag{\theeqlist}.\label{eqlist_2_4}
}

In equations~\eqref{eqlist_2_kappa1} and~\eqref{eqlist_2_kappa2}, $\kappa$ is defined by \eqref{eqkappa}, i.e., in each of these equations there are two cases. In equation \eqref{eqlist_2_fun} the function $\phi$ is defined by~\eqref{eqphi}, while $a$ and $c$ are arbitrary constants.
In the next list, corresponding transformations are presented.

\medskip
{ \centerline{{\bf List 2$'$.} Transformations of the second order.}}

\vspace{-5mm}

\spisok{
u=\frac{w_2}w,\stepcounter{trans}\tag{\thetrans}\label{trans_2_1}\\
u=w_2w_1w,\stepcounter{trans}\tag{\thetrans}\label{trans_2_3}\\
u=w_2w_1^{1+\kappa}w^\kappa,\stepcounter{trans}\tag{\thetrans}\label{trans_2_kappa1}\\
u=w_2w_1^{\kappa-1}w^{-\kappa},\stepcounter{trans}\tag{\thetrans}\label{trans_2_kappa2}\\
u=(w_2+w_1)(w_1+w),\stepcounter{trans}\tag{\thetrans}\label{trans_2_+}\\
u=(w_2-w_1+a)(w_1-w+a),\stepcounter{trans}\tag{\thetrans}\label{trans_2_-}\\
u=(\phi(w_2-w_1)-1)(\phi(w_1-w)-1),\stepcounter{trans}\tag{\thetrans}\label{trans_2_fun}\\
u=\frac{(w_2-w_1)(w_1-w)}{w_1}, \stepcounter{trans}\tag{\thetrans}\label{trans_2_sl}\\
u=w_2w_1(w+1),\stepcounter{trans}\tag{\thetrans.a}\label{trans_2_4a}\\
u=(w_2+1)w_1w. \tag{\thetrans.b}\label{trans_2_4b}
}

\begin{theorem}\label{ord2} If an equation \eqref{eqv} is transformed into the INB equation \eqref{INB} by a transformation~\eqref{tran} with $k=2$, then up to the autonomous point transformations~\eqref{point} and~\eqref{scale} it coincides to one of equations \eqref{eqlist_2_1}--\eqref{eqlist_2_4} of List~{\rm 2}. Equations \eqref{eqlist_2_1}--\eqref{eqlist_2_sl} are transformed into the INB equation by transformations \eqref{trans_2_1}--\eqref{trans_2_sl}, respectively. Equation~\eqref{eqlist_2_4} is transformed into the INB equation by any of transformations \eqref{trans_2_4a}--\eqref{trans_2_4b}.
\end{theorem}

Equations (\ref{eqlist_2_3}), (\ref{eqlist_2_+}), (\ref{eqlist_2_-}), (\ref{eqlist_2_sl}) together with corresponding transformations were discussed in \cite{GYL17,GYL18}, see also references therein. The important case~\eqref{eqlist_2_4} with transformations \eqref{trans_2_4a}--\eqref{trans_2_4b} has been found \cite[equation~(17.6.24)]{s03}.

All the transformations of List 2$'$ except for \eqref{trans_2_4a} and \eqref{trans_2_4b} are linearizable. For transformations \eqref{trans_2_1}--\eqref{trans_2_kappa2} one has
\begin{gather*} \hat u=\hat v_2+\alpha\hat v_1+\beta \hat v,\qquad u={\rm e}^{\hat u},\qquad v={\rm e}^{\hat v},\end{gather*} where $\alpha$, $\beta$ are some constants. Equations \eqref{eqlist_2_+}--\eqref{eqlist_2_fun} are transformed into~\eqref{eqlist1} by transformations
\begin{gather*} v=w_1+w,\qquad v=w_1-w+a,\qquad v=\phi(w_1-w)-1,\end{gather*} respectively, and therefore corresponding transformations \eqref{trans_2_+}--\eqref{trans_2_fun} are also linearizable as the compositions of linearizable transformations. Transformation~\eqref{trans_2_sl} is linearizable in a more complicated way, see \cite[Section~3]{GYL16}.

Both transformations \eqref{trans_2_4a} and \eqref{trans_2_4b} are of Miura type, see a comment in \cite[Section~2.1]{GYL16}.

\section{Modifications of the INB equation of the level 3}\label{third}
Here we classify all integrable modifications of the INB equation of the third level. First we give a list of equations.

\medskip

{ \centerline{{\bf List 3.} Modifications of the third level.}}
\vspace{-5mm}

\spisok{
\dot z=z\left(\frac{z_2}{z_{-1}}+\frac{z_1}{z_{-2}}\right)+cz,
\stepcounter{eqlist}\tag{\theeqlist}\label{eqlist_3_8}\\
\dot z=z_1z^3z_{-1}(z_2z_1-z_{-1}z_{-2}),
\stepcounter{eqlist}\tag{\theeqlist}\label{eqlist_3_pr}\\
\dot z=z\big(z_2z_1^\kappa z^{-1} z_{-1}^{-\kappa}-\kappa^2z_1z^\kappa z_{-1}^{-1}z_{-2}^{-\kappa}\big)+cz,
\stepcounter{eqlist}\tag{\theeqlist}\label{eqlist_3_2}\\
\dot z=(z_2-z+a)(z_1-z_{-1}+a)(z-z_{-2}+a)+c,
\stepcounter{eqlist}\tag{\theeqlist}\label{eqlist_3_1}\\
\dot z=-\big(T+1+T^{-1}\big)\frac{1}{(z_1-z+a)(z-z_{-1}+a)}+c,
\stepcounter{eqlist}\tag{\theeqlist}\label{eqlist_3_aa}\\
\dot z=-\big(T-\kappa+\kappa^2 T^{-1}\big)\frac{1}{(z_1+\kappa z)(z+\kappa z_{-1})},
\stepcounter{eqlist}\tag{\theeqlist}\label{eqlist_3_aa1}\\
\dot z=w_1w+ww_{-1}+w_{-1}w_{-2}-w_1ww_{-1}-ww_{-1}w_{-2}+c,
\stepcounter{eqlist}\tag{\theeqlist}\label{eqlist_3_7}\\
\hphantom{\dot z=}{} \ w=\frac{1}{1-\phi(z_1-z)},\nonumber\\
\dot z=z(z_1z-1)(zz_{-1}-1)(z_2z_1-z_{-1}z_{-2}),
\stepcounter{eqlist}\tag{\theeqlist}\label{eqlist_3_3}\\
\dot z=\frac{(az_1-z)(az-z_{-1})(az_2z_{-2}+az_1z_{-1}-2z_1z_{-2})}{z_1z_{-1}z_{-2}}+cz,\qquad a\neq0,
\stepcounter{eqlist}\tag{\theeqlist}\label{eqlist_3_4}\\
\frac{\dot z}{a^2}=\frac{z_2(az+z_{-1})}{z_{-1}}+\frac{z_1z}{z_{-1}}+\frac{(az_1+z)z}{z_{-2}}+cz,
\stepcounter{eqlist}\tag{\theeqlist}\label{eqlist_3_5}\\
\dot z=-z\big(T+1+T^{-1}\big)\frac{zz_{-1}}{(az_1-z)(az-z_{-1})}+cz,\qquad a\neq0,
\stepcounter{eqlist}\tag{\theeqlist}\label{eqlist_3_81}\\
\dot z=-z\big(T-\kappa+\kappa^2T^{-1}\big)\frac{1}{(z_1z^\kappa-1)(zz_{-1}^\kappa-1)},
\stepcounter{eqlist}\tag{\theeqlist}\label{eqlist_3_71}\\
\dot z=\theta(z_1-z)\theta(z-z_{-1}) [\theta(z_{2}-z_1)+\theta(z_{-1}-z_{-2})+a ]\nonumber\\
\hphantom{\dot z=}{} +a[\theta(z_2-z_1)\theta(z_1-z)+\theta(z-z_{-1})\theta(z_{-1}-z_{-2})]+c,
\stepcounter{eqlist}\tag{\theeqlist}\label{eqlist_3_6}\\
\theta'=\frac{\theta(\theta+1)}{\theta+a},\qquad a\neq 0,\qquad a\neq1,\nonumber\\
\dot z=4\big(z^2-1\big)(2+(z_{-1}-1)T^{-1}-(z_1+1)T)\Omega,
\stepcounter{eqlist}\tag{\theeqlist}\label{eqlist_3_Mi}\\
\Omega=\frac{1}{[(z_1+1)(z-1)+4][(z+1)(z_{-1}-1)+4]},\nonumber\\
\dot z=-z\big(T+1+T^{-1}\big)\frac{z}{(z_1-z)(z-z_{-1})}.
\stepcounter{eqlist}\tag{\theeqlist}\label{eqlist_3_N}
}

Here $\kappa$ is defined by \eqref{eqkappa}, i.e., for each such an equation with a dependence on~$\kappa$, there are two cases. Besides, $a$ and $c$ are arbitrary constants, while in equation~\eqref{eqlist_3_7} the function $\phi$ is defined by~\eqref{eqphi}. Let us list now the corresponding transformations.

\medskip
{ \centerline{{\bf List 3$'$.} Transformations of the third order.}}

\vspace{-5mm}

\spisok{
u=\frac{z_3}z,
\stepcounter{trans}\tag{\thetrans}\label{trans_3_8}\\
u=z_3z_2^2z_1^2z,
\stepcounter{trans}\tag{\thetrans}\label{trans_3_pr}\\
u=z_3z_2^{\kappa}z_1^{-1}z^{-\kappa},
\stepcounter{trans}\tag{\thetrans}\label{trans_3_2}\\
u=(z_3-z_1+a)(z_2-z+a),
\stepcounter{trans}\tag{\thetrans}\label{trans_3_1}\\
u=\frac{1}{(z_3-z_2+a)(z_2-z_1+a)(z_1-z+a)},
\stepcounter{trans}\tag{\thetrans}\label{trans_3_aa}\\
u=\frac{1}{(z_3+\kappa z_2)(z_2+\kappa z_1)(z_1+\kappa z)},
\stepcounter{trans}\tag{\thetrans}\label{trans_3_aa1}\\
u=\frac{1}{(1-\phi(z_3-z_2))(1-\phi(z_2-z_1))(1-\phi(z_1-z))},
\stepcounter{trans}\tag{\thetrans}\label{trans_3_7}\\
u=(z_3z_2-1)(z_2z_1-1)z_1z,
\stepcounter{trans}\tag{\thetrans.a}\label{trans_3_3_a}\\
u=z_3z_2(z_2z_1-1)(z_1z-1),
\tag{\thetrans.b}\label{trans_3_3_b}\\
u=a\frac{(az_3-z_2)(az_2-z_1)}{z_2z},
\stepcounter{trans}\tag{\thetrans.a}\label{trans_3_4_a}\\
u=a\frac{z_3(az_2-z_1)(az_1-z)}{z_2z_1z},
\tag{\thetrans.b}\label{trans_3_4_b}\\
u=a^2\frac{z_3(az_1+z)}{z_1z},
\stepcounter{trans}\tag{\thetrans.a}\label{trans_3_5_a}\\
u=a^2\frac{az_3+z_2}{z},
\tag{\thetrans.b}\label{trans_3_5_b}\\
u=\frac{az_2z_1^2}{(az_3-z_2)(az_2-z_1)(az_1-z)},
\stepcounter{trans}\tag{\thetrans.a}\label{trans_3_81_a}\\
u=\frac{az_3z_1z}{(az_3-z_2)(az_2-z_1)(az_1-z)},
\tag{\thetrans.b}\label{trans_3_81_b}\\
u=\frac{z_1z^\kappa}{(z_3z_2^\kappa-1)(z_2z_1^\kappa-1)(z_1z^\kappa-1)},
\stepcounter{trans}\tag{\thetrans.a}\label{trans_3_71_a}\\
u=\frac{z_3z_2^\kappa}{(z_3z_2^\kappa-1)(z_2z_1^\kappa-1)(z_1z^\kappa-1)},
\tag{\thetrans.b}\label{trans_3_71_b}\\
u=\theta(z_3-z_2)\theta(z_2-z_1)(\theta(z_1-z)+1),
\stepcounter{trans}\tag{\thetrans.a}\label{trans_3_6_a}\\
u=(\theta(z_3-z_2)+1)\theta(z_2-z_1)\theta(z_1-z),
\tag{\thetrans.b}\label{trans_3_6_b}\\
u=\frac{4(z_2+1)\big(z_1^2-1\big)(z-1)}{[(z_3+1)(z_2-1)+4][(z_2+1)(z_1-1)+4][(z_1+1)(z-1)+4]},
\stepcounter{trans}\tag{\thetrans}\label{trans_3_Mi}\\
u=\frac{z_2z_1}{(z_3-z_2)(z_2-z_1)(z_1-z)}.
\stepcounter{trans}\tag{\thetrans}\label{trans_3_N}
}

\begin{theorem}\label{ord3} If an equation \eqref{eqv} is transformed into the INB equation \eqref{INB} by a transformation of the form \eqref{tran} with $k=3$, then up to autonomous point transformations~\eqref{point} and~\eqref{scale} it coincides with one of equations \eqref{eqlist_3_8}--\eqref{eqlist_3_N} of the List~{\rm 3}. Equations \eqref{eqlist_3_8}--\eqref{eqlist_3_N} are transformed into the INB equation by transformations \eqref{trans_3_8}--\eqref{trans_3_N}, respectively. For each of equations \eqref{eqlist_3_3}--\eqref{eqlist_3_6}, there are two transformations.
\end{theorem}

Equations (\ref{eqlist_3_pr}), (\ref{eqlist_3_3}) together with corresponding transformations were discussed in \cite{GYL17,GYL18}, see also references therein. Equation \eqref{eqlist_3_1} with $a=c=0$ has been presented \cite{GYL16}. Transformations \eqref{trans_3_8}--\eqref{trans_3_2} are linearized as follows
\begin{gather}\hat u=\hat v_3+\alpha\hat v_2+\beta \hat v_1+\gamma \hat v,\qquad u={\rm e}^{\hat u},\qquad v={\rm e}^{\hat v},\label{lin}\end{gather} where $\alpha$, $\beta$, $\gamma$ are some constants.

For further discussion, we will need some auxiliary transformations.

\medskip
{ \centerline{{\bf List 3$''$.} Auxiliary transformations.}}
\vspace{-7mm}

\spisok{
\eqref{eqlist_3_1}\to\eqref{eqlist1}\colon \qquad v=z_2-z+a, \tag{T$^*18$}\label{zz18}
\\\eqref{eqlist_3_aa}\to\eqref{eqlist_2_3}\colon\qquad w=\frac{1}{z_1-z+a},\tag{T$^*19$}
\\\eqref{eqlist_3_aa1}\to\eqref{eqlist_2_3}\colon\qquad w=\frac{1}{z_1+\kappa z},\tag{T$^*20$}
\\\eqref{eqlist_3_7} \to\eqref{eqlist_2_3}\colon \qquad w=\frac{1}{1-\phi(z_1-z)},\tag{T$^*21$}\label{zz20}
\\\eqref{eqlist_3_3}\to \eqref{eqlist_2_4}\colon \qquad w=z_1z-1,\label{z21}\tag{T$^*22$}
\\\eqref{eqlist_3_4}\to \eqref{eqlist_2_4}\colon \qquad w=a\frac{z_1}{z}-1,\tag{T$^*23$}
\\\eqref{eqlist_3_5}\to \eqref{eqlist_2_4}\colon \qquad w=a\frac{z_1}z,\tag{T$^*24$}
\\\eqref{eqlist_3_81}\to \eqref{eqlist_2_4}\colon \qquad w=\frac{z}{az_1-z},\tag{T$^*25$}
\\\eqref{eqlist_3_71}\to \eqref{eqlist_2_4}\colon \qquad w=\frac{1}{z_1 z^{\kappa}-1},\tag{T$^*26$}
\\\eqref{eqlist_3_6}\to \eqref{eqlist_2_4}\colon \qquad w=\theta(z_1-z),\label{z27}\tag{T$^*27$}
\\\eqref{eqlist_3_Mi}\to \eqref{eqlist_2_4}\colon \qquad w=-\frac{(z_1+1)(z-1)}{(z_1+1)(z-1)+4}.\tag{T$^*28$}\label{zz27}}

Transformations \eqref{trans_3_1}--\eqref{trans_3_7} are compositions of linearizable transformations \eqref{zz18}--\eqref{zz20} and of one of linearizable transformations \eqref{trans1} or \eqref{trans_2_3} and therefore are linearizable too.

{\sloppy All transformations \eqref{z21}--\eqref{z27} are linearizable. In the case of equations \eqref{eqlist_3_3}--\eqref{eqlist_3_6}, transformations \eqref{trans_3_3_a}--\eqref{trans_3_6_a} are compositions of transformations \eqref{z21}--\eqref{z27} and of~\eqref{trans_2_4a}, while \eqref{trans_3_3_b}--\eqref{trans_3_6_b} are compositions of transformations \eqref{z21}--\eqref{z27} and of~\eqref{trans_2_4b}. As we have mentioned above, transformations~\eqref{trans_2_4a} and~\eqref{trans_2_4b} are of Miura type, and for this reason all the transformations associated with (\ref{eqlist_3_3})--(\ref{eqlist_3_6}) are compositions of linearizable and Miura type transformations.

}

Transformation \eqref{trans_3_Mi} is a composition of two Miura type transformations, see details in Section~\ref{secN2} below.
Transformation \eqref{trans_3_N} is linearized in a non-standard way, see Section~\ref{secN1}.

\section{The most interesting examples}\label{secN}

Here we discuss in detail equations \eqref{eqlist_3_N} and \eqref{eqlist_3_Mi} and corresponding transformations \eqref{trans_3_N} and \eqref{trans_3_Mi}.

\subsection{Equation (\ref{eqlist_3_N})}\label{secN1}

Transformation \eqref{trans_3_N} corresponding to \eqref{eqlist_3_N} is quite similar to \eqref{trans_2_sl} which is discussed in detail in \cite[Section 3]{GYL16}. As it will be shown below, unlike the most of linearizable transformations presented in this paper, it is a composition of linearizable transformations in different directions.

In this case we need one auxiliary equation
\begin{gather}
\dot y=\frac{1}{(y_2-y_{-1})(y_1-y_{-2})}\label{eqy1}\end{gather}
and two obviously linearizable transformations
\begin{gather}
\eqref{eqy1}\to\eqref{eqlist_3_N}\colon \qquad z=-y_{2}y_{1}y\label{ztoy},\\
\eqref{eqy1}\to\eqref{eqlist_2_3}\colon \qquad w=\frac{1}{y-y_3}\label{wtoy}.
\end{gather}
Transformation \eqref{trans_3_N} can be decomposed as it is shown in the following diagram
\begin{gather}\begin{diagram} \label{pippo}
\eqref{eqlist_3_N} &
\rTo{}{\eqref{trans_3_N}} & \eqref{INB}
\\
\uTo{\eqref{ztoy}} & & \uTo{}{\eqref{trans_2_3}} \\
 \eqref{eqy1}& \rTo{\eqref{wtoy}} & \eqref{eqlist_2_3}
\end{diagram}\end{gather}

As we have mentioned in Section~\ref{theo}, transformation~\eqref{trans_3_N} allows one to construct conservation laws for the modified equation~\eqref{eqlist_3_N} in an explicit way, as this transformation is of the form \eqref{tran}. The construction of generalized symmetries is a more difficult problem. However, it is simplified in the case of linearizable transformations, as it is reduced to the use of linear transformations like~\eqref{lin}. The construction of modified equations by using linearizable transformations is discussed in \cite[Appendix~A]{GYL17}. Using that theory, we construct here a generalized symmetry for~\eqref{eqlist_3_N}.

The decomposition shown in diagram~\eqref{pippo} allows one to construct the generalized symmetry for~\eqref{eqlist_3_N} by using a known symmetry for the INB equation \eqref{INB}. It is even easier to use a~sym\-met\-ry, presented in~\cite{ztof91}, for its well-known modification~\eqref{eqlist_2_3}. At first we construct a~generalized symmetry for~\eqref{eqy1}, which is of the form
\begin{gather*} y_\tau=\frac{1}{(y_2-y_{-1})(y_1-y_{-2})}\big(T^3+T^2+T+1\big)\frac{1}{(y_1-y_{-2})(y-y_{-3})(y_{-1}-y_{-4})}.\end{gather*}
Now we can find a generalized symmetry for \eqref{eqlist_3_N}, which reads
\begin{gather*} z_\tau =z\big(T+1+T^{-1}\big)\frac{z\Xi}{(z_1-z)(z-z_{-1})},\\ \Xi =\big(T^3+T^2+T+1\big)\frac{z_{-1}z_{-2}}{(z-z_{-1})(z_{-1}-z_{-2})(z_{-2}-z_{-3})}.\end{gather*}

Equation \eqref{eqy1} exemplifies a modification of the INB equation \eqref{INB} of the highest possible fifth level. Corresponding transformation of \eqref{eqy1} into \eqref{INB} is the composition of transformations shown in diagram~\eqref{pippo} and it reads
\begin{gather*}u=-\frac{1}{y_5-y_2}\frac{1}{y_4-y_1}\frac{1}{y_3-y}.\end{gather*}

\subsection{Equation (\ref{eqlist_3_Mi})}\label{secN2}
Transformation \eqref{zz27}, unlike all other transformations of List $3''$, is of Miura type. Rewriting it in the form
\begin{gather} \label{tilw}\widetilde w=-\frac{4w}{w+1}=(z_1+1)(z-1),\end{gather} we see that it is the standard Miura transformation of $z$ to $\widetilde w$, cf.~\eqref{miura}.

Transformation \eqref{trans_3_Mi} is a composition of transformations~\eqref{zz27} and~\eqref{trans_2_4b}. The second natural composition of transformations~\eqref{zz27} and~\eqref{trans_2_4a} is not written down in List~$3'$, as it comes from~\eqref{trans_3_Mi} by using the point transformation
\begin{gather} \tilde z=\varsigma(z)=\frac{z+3}{1-z}.\label{avtoz}\end{gather} It is the auto-transformation of equation~\eqref{eqlist_3_Mi}. It is interesting that \begin{gather} \label{avtoz2}\varsigma^2(z)=\varsigma(\varsigma(z))=\frac{z-3}{z+1},\qquad \varsigma^3(z)=z,\end{gather} i.e., there are two different auto-transformations for equation~\eqref{eqlist_3_Mi}.

Using the Miura type transformation \eqref{zz27} and corresponding modification \eqref{eqlist_3_Mi}, let us construct in this section one more non-standard transformation and an integrable completely discrete quad-equation.

{\bf Non-standard transformation.} A Miura-like transformation can be constructed with the help of a theory developed in \cite{y93, y94}.

Using transformation \eqref{zz27}, let us construct a relation
\begin{gather}
w=-\frac{(z_1+1)(z-1)}{(z_1+1)(z-1)+4}=-\frac{(\hat z_1+1)(\hat z-1)}{(\hat z_1+1)(\hat z-1)+4}\label{zzhat}\end{gather} for two solutions $z$, $\hat z$ of equation \eqref{eqlist_3_Mi}. This relation is compatible with equation~\eqref{eqlist_3_Mi}. It is simplified in terms of the function $\widetilde w$ given by~\eqref{tilw}
\begin{gather*} \widetilde w=(z_1+1)(z-1)=(\hat z_1+1)(\hat z-1).\end{gather*} The latter can be rewritten as
\begin{gather*} \frac{z-1}{\hat z-1}=\frac{\hat z_1+1}{z_1+1}=-y_1,\end{gather*}
where $y_n$ is a new unknown function. We get a transformation
\begin{gather} y=-\frac{\hat z+1}{z+1},\qquad y_1=-\frac{z-1}{\hat z-1}\label{yz},\end{gather} which is invertible on the solutions of~\eqref{zzhat}
\begin{subequations}\label{zy1+zy2}
\begin{gather}
z=\frac{y_1y+2y_1+1}{1-y_1y},\label{zy1}\\
 \hat z=\frac{y_1y+2y+1}{y_1y-1}.\label{zy2}
 \end{gather}
\end{subequations} This transformation allows us to rewrite equation~\eqref{eqlist_3_Mi} in terms of~$y_n$.

Differentiating the first of relations \eqref{yz} with respect to the time in virtue of \eqref{eqlist_3_Mi} and substituting the functions (\ref{zy1+zy2}), we get an equation
\begin{gather} \label{eqy}\dot y=\frac{y(y+1)(y_1y-1)(yy_{-1}-1)(y_2y_1-y_{-1}y_{-2})}{(y_2y_1y+1)(y_1yy_{-1}+1)(yy_{-1}y_{-2}+1)}.\end{gather} This equation coincides with \cite[equation~(5.11)]{x18}, where $\chi=-1$, and it is equivalent to~\cite[equation~(59b)]{mx13} up to a point transformation. Equation~\eqref{eqy} is transformed into \eqref{eqlist_3_Mi} by any of transformations (\ref{zy1+zy2}). It should be remark that these two transformations are equivalent up to the point auto-transformation $y_n\to 1/y_n$ of equation~\eqref{eqy}. So, we have derived modification~\eqref{eqy} of equation~\eqref{eqlist_3_Mi} and two equivalent transformations that look like Miura type transformations.

Let us analyze in more detail transformation \eqref{zy1}. It can be rewritten as a discrete Riccati equation for the function~$y_n$
\begin{gather} (1+z)y_1y+2y_1+1-z=0.\label{ricy}\end{gather} In accordance with~\cite{GYL16} it should be the Miura type transformation. However, a particular solution of this equation can easily be found: $y_n^*\equiv -1.$ Therefore, unlike the case of general position, it can be linearized in an explicit way and can be solved by quadrature.

In fact, using this particular solution, let us change $y_n=y_n^*+1/\tilde{y}_n$ in order to get a non-autonomous linear equation
\begin{gather} \tilde y_1-\frac{1-z}{1+z}\tilde y-1=0.\label{eqyl}\end{gather} Substituting the relation
\begin{gather} \frac{1-z}{1+z}=\frac{x-x_{-1}}{x_1-x}\label{zx}\end{gather} in terms of a new function~$x_n$, we rewrite~\eqref{eqyl} as the equation
\begin{gather*} (T-1)[(x-x_{-1})\tilde y-x]=0,\end{gather*} which has the obvious solution
\begin{gather*}\tilde y=\frac{x}{x-x_{-1}}.\end{gather*}

As a result we have $y=-\frac{x_{-1}}{x}$. Starting from a given function $z_n$, we can find $x_n$ from~\eqref{zx} by using two discrete integrations, as
\begin{gather*} x-x_{-1}=r,\qquad \frac{r}{r_1}=\frac{1-z}{1+z}.\end{gather*}
So, the Riccati equation \eqref{ricy} is degenerate in the sense that it is solved by quadrature.

These results can be reformulated in terms of linearizable transformations and modifications of the INB equation. We are led to the following picture
\begin{gather}\begin{diagram} \label{yz*}
\eqref{eqy} &
\rTo{}{\eqref{zy1}} & \eqref{eqlist_3_Mi}
\\
\uTo{y=-\dfrac{x_{-1}}{x}} & & \uTo{}{\dfrac{1-z}{1+z}=\dfrac{r}{r_1}}\\%\eqref{trans_2_3}} \\
\eqref{eqx} & \rTo{r=x-x_{-1}} & \eqref{eqr} %\eqref{eqlist_2_3}
\end{diagram}\end{gather}
where
\begin{gather}\dot x=\frac{(x_2-x)(x_1-x_{-1})(x-x_{-2})}{(x_2-x_{-1})(x_1-x_{-2})}\label{eqx},\\
\dot r=\frac{r(r_2+r_1-r_{-1}-r_{-2})(r_1+r)(r+r_{-1})}{(r_2+r_1+r)(r_1+r+r_{-1})(r+r_{-1}+r_{-2})}.\label{eqr}
\end{gather}
We have that \eqref{zy1} is a complete analogue of transformations \eqref{trans_3_N} and \eqref{trans_2_sl}, i.e.,
it is a~composition of linearizable transformations in different directions. Therefore it is not of Miura type, but it is the linearizable transformation in accordance with the terminology of~\cite{GYL16}.

Equation \eqref{eqx} has been found in \cite[equation~(2.23)]{pn96}, \eqref{eqr} coincides with \cite[equation~(3.1b)]{shl14}, while~\eqref{eqy} has been discussed above. New objects in this diagram are equation~\eqref{eqlist_3_Mi} and transformation~\eqref{zy1}. Equations (\ref{eqy}), (\ref{eqx}), (\ref{eqr}) exemplify modifications of the INB equation \eqref{INB} of the fourth and fifth levels. Corresponding transformations into~\eqref{INB} are obtained as compositions of transformation \eqref{trans_3_Mi} and transformations shown in diagram~\eqref{yz*}. Those transformations are too cumbersome to be written here.

{\bf Integrable quad-equation.} It is known that differential-difference equations like~\eqref{eqv} and discrete equations on the square lattice (quad-equations) are closely related. Differential-dif\-fe\-ren\-ce equations define the generalized symmetries for quad-equations, while quad-equations can be interpreted as the B\"acklund transformations for differential-difference equations~\cite{lpsy08}. The problem of construction of a compatible quad-equation for a given differential-difference equation was solved in \cite{ggy18, gy15}. In~\cite{ggy18} the five-point differential-difference equations~\eqref{eqv} were considered.

Let us find an integrable quad-equation compatible with equation \eqref{eqlist_3_Mi} by using a different method. We use here a theory developed in \cite{s10,y90}. From \cite{mx13} we know that equation \eqref{eqlist_2_4}, expressed in the form
\begin{gather}\dot w_{n,m}=w_{n,m}(w_{n,m}+1)(w_{n+2,m}w_{n+1,m}-w_{n-1,m}w_{n-2,m}),\label{eqwnm}\end{gather} is the generalized symmetry of quad-equation
\begin{gather} w_{n+1,m+1}(w_{n,m+1}+w_{n,m}+1)+w_{n,m}(w_{n+1,m}+1)=0.\label{dis_w}\end{gather}

Let us rewrite it in the form
\begin{gather*} 2\frac{w_{n+1,m}+1}{w_{n+1,m+1}}+1=-2\frac{w_{n,m+1}+1}{w_{n,m}}-1.\end{gather*} This allows us to introduce a new function $z_{n,m}$, so that
\begin{gather}z_{n+1,m}=-2\frac{w_{n,m+1}+1}{w_{n,m}}-1, \qquad z_{n,m}=2\frac{w_{n,m}+1}{w_{n,m+1}}+1.\label{zw}\end{gather} It is obvious that
\begin{gather} w_{n,m}=-\frac{2(z_{n,m}+1)}{(z_{n+1,m}+1)(z_{n,m}-1)+4},\qquad w_{n,m+1}=\frac{2(z_{n+1,m}-1)}{(z_{n+1,m}+1)(z_{n,m}-1)+4}.\label{wz}\end{gather}
Rewriting these relations at the same point $w_{n,m+1}$, we get an equation for $z_{n,m}$
\begin{gather}(z_{n+1,m+1}-1)(z_{n+1,m}-1)(z_{n,m+1}-1)+(z_{n,m}+1)(z_{n+1,m}+1)(z_{n,m+1}+1)=0.\label{dis_z}\end{gather} So, this new discrete equation is obtained by transformation \eqref{zw} which is invertible on the solutions of quad-equation~\eqref{dis_w}.

Transformation (\ref{zw}), (\ref{wz}) allows one to rewrite the generalized symmetries of quad-equa\-tion~\eqref{dis_w}. Differentiating the second of relations~\eqref{zw} with respect to the time in virtue of~\eqref{eqwnm} and substituting the functions~\eqref{wz}, we get a generalized symmetry for~\eqref{dis_z}.
Let us denote the right hand side of equation~\eqref{eqlist_3_Mi} by
\begin{gather*}\Theta(z_{n+2},z_{n+1},z_{n},z_{n-1},z_{n-2}).\end{gather*} Then the resulting generalized symmetry is of the form
\begin{gather}\dot z_{n,m}=\Theta(z_{n+2,m},z_{n+1,m},z_{n,m},z_{n-1,m},z_{n-2,m}),\label{eqzn}\end{gather} i.e., it is defined by equation~\eqref{eqlist_3_Mi}. Due to the invariance of quad-equation \eqref{dis_z} under the change of discrete variables $n\leftrightarrow m$, a generalized symmetry in the other direction~$m$ reads
\begin{gather} z'_{n,m}=\Theta(z_{n,m+2},z_{n,m+1},z_{n,m},z_{n,m-1},z_{n,m-2}).\label{eqzm}\end{gather}

A generalized symmetry of \eqref{dis_w} in the $m$-direction is presented in \cite{mx13} and it is equivalent to symmetry \eqref{eqzm} up to the invertible transformation (\ref{zw}), (\ref{wz}). Formulae \eqref{wz} define two Miura type transformations which transform \eqref{eqzn} into \eqref{eqwnm}. Both of them are equivalent to the discrete Miura transformation~\eqref{zz27} up to auto-transformations given in (\ref{avtoz}), (\ref{avtoz2}).

\section{Conclusion}
In this paper we have enumerated all integrable modifications of the form~\eqref{eqvgen} of the INB equation~\eqref{INB} of the levels $k=1,2,3$, which are related to it by non-invertible transformations of the form~\eqref{tran_g}. Resulting Lists 1--3 contain 29 equations, see more precise formulations in Theorems~\ref{ord1}--\ref{ord3}. As far as we know, a classification problem of this kind is solved for the first time in the discrete case.

Corresponding discrete non-invertible transformations are presented in Lists $1'$--$3'$. We have analyzed those transformations to understand their nature. As a result we have shown that the most of them are linearizable.

We have also proved that the orders of possible transformations in this problem are restricted by the number five, see Theorem~\ref{order}. This estimate is accurate in the sense that there exist transformations for each order~$k$: $1\leq k\leq 5$.
Transformations of the orders $k=1,2,3$ have been completely enumerated in Sections~\ref{low} and~\ref{third}. Examples of the orders $k=4,5$ are presented in Section~\ref{secN}. The complete classification of transformations of the fourth and fifth orders is left for a future work.

As a result of the classification we obtain a number of new integrable equations and non-invertible discrete transformations. Among new transformations we would like to highlight transformation \eqref{zz27} of Miura type and two transformations~\eqref{trans_3_N} and \eqref{zy1} that are li\-nearizable in a non-standard way. Among equations we note two five-point differential-difference equations \eqref{eqlist_3_Mi}, \eqref{eqlist_3_N} and one quad-equation~\eqref{dis_z} compatible with~\eqref{eqlist_3_Mi}.

\subsection*{Acknowledgments}
RIY gratefully acknowledges the financial support from Russian Science Foundation grant (pro\-ject 15-11-20007).

\pdfbookmark[1]{References}{ref}
\LastPageEnding

\end{document}